\documentclass[]{amsart}
\usepackage{amsmath,amssymb,amsfonts,amsthm,bbm}

\usepackage{lmodern}
\usepackage[T1]{fontenc}

\usepackage[pdftex]{graphicx}
\usepackage{caption}
\usepackage{subfig}
\usepackage{color}
\usepackage{float}
\restylefloat{figure}

\newtheorem{theorem}{Theorem}[section]
\newtheorem{proposition}[theorem]{Proposition}
\newtheorem{lemma}[theorem]{Lemma}

\newtheorem{remark}{Remark}
\newtheorem{example}{Example}

\newcommand{\HH}{\mathcal{H}}
\newcommand{\KK}{\mathcal{K}}
\newcommand{\D}{\mathcal{D}}
\renewcommand{\S}{\mathcal{S}}

\newcommand{\A}{\mathcal{A}}
\renewcommand{\L}{\mathcal{L}}

\newcommand{\U}{\mathcal{U}}

\renewcommand{\u}{\mathfrak{u}}

\newcommand{\T}{\operatorname{T}\!}
\renewcommand{\H}{\operatorname{H}\!}
\newcommand{\V}{\operatorname{V}\!}
\newcommand{\Ker}{\operatorname{Ker}}

\newcommand{\ad}{\operatorname{ad}}

\newcommand{\diag}{\operatorname{diag}}

\renewcommand{\phi}{\varphi}

\newcommand{\1}{{\mathbf 1}}

\newcommand{\dt}{\operatorname{d}\!t}

\newcommand{\dd}[1]{\frac{\operatorname{d}}{\operatorname{d}\!#1}}
\newcommand{\length}[1]{\operatorname{Length}[#1]}

\renewcommand{\Re}{\operatorname{Re}}
\renewcommand{\Im}{\operatorname{Im}}

\newcommand{\ket}[1]{|{#1}\rangle}
\newcommand{\bra}[1]{\langle #1 |}
\newcommand{\braket}[2]{\langle #1|#2\rangle}
\newcommand{\ketbra}[2]{|#1\rangle\langle #2|}
\newcommand{\dist}[2]{\operatorname{Dist}(#1,#2)}

\newcommand{\angleB}[2]{\mathcal{L}_{\textit{B}}(#1,#2)}

\newcommand{\ptexp}{\exp_{+}}
\newcommand{\ntexp}{\exp_{-}}

\newcommand{\obs}[1]{\hat{#1}}
\newcommand{\Tr}{\operatorname{Tr}}

\begin{document}
\title[Quantum speed limits and optimal Hamiltonians]{Quantum speed limits and optimal Hamiltonians for driven systems in mixed states}
\author{Ole Andersson}
\email{olehandersson@gmail.com}
\author{Hoshang Heydari}
\thanks{The second author acknowledges the financial support from the Swedish Research Council (VR), grant number
2008-5227.}
\address{Department of Physics, Stockholm University, 10691 Stockholm, Sweden}
\date{\today}

\begin{abstract}
Inequalities of Mandelstam-Tamm and Margolus-Levitin type provide lower bounds on the time it takes for a quantum system to evolve from one state into another. Knowledge of such bounds, called quantum speed limits, is of utmost importance in virtually all areas of physics, where determination of the minimum time required for a quantum process is of interest.
Most Mandelstam-Tamm and Margolus-Levitin inequalities found in the literature have been derived from growth estimates for the Bures length, which is a statistical distance measure. In this paper we derive such inequalities by differential geometric methods, and we compare the obtained quantum speed limits with those involving the Bures length. We also characterize the Hamiltonians which optimize the evolution time for generic finite-level quantum systems.
\end{abstract}

\maketitle

\section{Introduction}
The fundamental problem of determining the minimum time required to perform a quantum process, and the dual problem of designing time-optimal Hamiltonians, have recently attracted much attention because of their significance in several modern applications of quantum mechanics. These include quantum metrology \cite{Zwierz_etal2010,Giovannetti_etal2011,Zwierz_etal2012}, quantum computation and information \cite{Bekenstein1981,Lloyd2000,Giovannetti_etal2003}, and optimal control theory \cite{Caneva_etal2009,Bason_etal2012,Poggi_etal2013,Hegerfeldt2013,Poggi_etal2013EPL}.
Limits on the minimal evolution time also play a role in cosmology \cite{Bekenstein1981}, and quantum thermodynamics \cite{Deffner_etal2010}.

In this paper, we use the terminology introduced by Margolus and Levitin \cite{Margolus_etal1998} and call lower bounds on the time it takes for a quantum system to evolve from one state into another \emph{quantum speed limits}. 
More specifically, we refer to lower bounds that involve the energy uncertainty as Mandelstam-Tamm ({\sc MT}) quantum speed limits.
Bhattacharyya \cite{Bhattacharyya1983} was one of the first to derive an {\sc MT} quantum speed limit (see also \cite{Fleming1973}). 
Assuming Mandelstam and Tamm's uncertainty relation, he showed through accurate estimates of the rate of change of the ``quantum non-decay probability'' that it takes at least the time $\pi\hbar/2\Delta E$ for a system to evolve between two orthogonal pure states.
A few years later,  Anandan and Aharonov \cite{Anandan_etal1990} confirmed Bhattacharyya's result.
But more importantly, they revealed the geometric nature of the {\sc MT} quantum speed limits, showing that $1/\hbar$ times the path integral of the energy uncertainty of a unitarily evolving pure state equals the Fubini-Study length of the curve traced out by the state.

Few evolution time estimates for quantum systems in mixed states have been derived by differential geometric methods, despite the differential geometric approach in \cite{Anandan_etal1990}. 
Until now, most {\sc MT} quantum speed limits for mixed states have been obtained from estimates of the growth of the Bures length \cite{Bures1969}, which is a statistical distance measure. Soon after the publication of \cite{Anandan_etal1990}, Uhlmann \cite{Uhlmann1992Energy} showed that the time it takes for a system to evolve from one mixed state into another is bounded from below by $\hbar$ times the fraction of the Bures length between the states and the energy uncertainty of the system. Uhlmann's result has been verified in several publications, e.g. \cite{Jones_etal2010,Zwierz2012,Giovannetti_etal2003(2),Giovannetti_etal2003proc,Deffner_etal2013}.

In this paper we develop Anandan and Aharonov's approach. We use symplectic reduction to construct a principal fiber bundle over a general space of isospectral (i.e. unitarily equivalent) mixed quantum states. The bundle, which generalizes the Hopf bundle for pure states, gives rise in a canonical way to a Riemannian metric and a symplectic form on the space of isospectral mixed states. 
Using these we then derive an {\sc MT} quantum speed limit for unitarily driven quantum systems, which proves to be sharper than the Uhlmann limit.

The speed of a quantum evolution depends also on the system's energy resources.
Margolus and Levitin \cite{Margolus_etal1998} showed that the time it takes for a non-driven system to unitarily evolve between two orthogonal pure states is bounded from below by a factor that is inversely proportional to the energy of the system.
Generalizing Margolus and Levitin's result to systems in mixed states has proven to be quite difficult.
Giovannetti \emph{et al.} \cite{Giovannetti_etal2003(2),Giovannetti_etal2003proc} have derived an ``implicit'' Margolus-Levitin ({\sc ML}) quantum speed limit for non-driven systems, and Deffner and Lutz \cite{Deffner_etal2013} used the methods put forward by Jones and Kok \cite{Jones_etal2010} to derive an estimate of Margolus-Levitin type for driven systems.
The present paper also contains a generalization of Margolus and Levitin's quantum speed limit to driven systems in mixed states. The speed limit is different from that of Deffner and Lutz, and it reduces to a greater speed limit than theirs for systems with time-independent Hamiltonians.

The paper is organized as follows. In Section \ref{red pur bund}, we set up the geometric framework, and introduce most of the notation we use in the rest of the paper.
In Sections \ref{evolution time} and \ref{margolus} we derive an {\sc MT} quantum speed limit and an {\sc ML} quantum speed limit, respectively, for the time it takes to unitarily run a quantum system from one mixed state into another. Section \ref{optimal hamiltonians} contains a characterization of the Hamiltonians that optimize evolution time for finite-level systems in generic mixed states, as well as an example of a system for which the {\sc MT} quantum speed limit derived in Section \ref{evolution time} is greater than the corresponding limit involving the Bures length. The paper ends with a conclusion.

\subsection{Conventions}
Evolving quantum states will be represented by curves of density operators.
The curves are assumed to be defined on an unspecified interval $0\leq t\leq \tau$, and the final time $\tau$ will be referred to as the evolution time.
Compositions of linear mappings will be written as concatenations.
By a ``function'' we mean a real-valued smooth function, and by a ``functional'' we mean a real-valued linear function. 
We use the same notation, namely $\1$, for the identity map of every space.

\section{Reduced purification bundles}\label{red pur bund}
This paper concerns quantum systems in mixed states which evolve unitarily.
The systems will be modeled on a Hilbert space $\HH$ and their states will be represented by density operators.
We write $\D(\HH)$ for the space of density operators on $\HH$, and $\D_k(\HH)$ for the space of density operators on $\HH$ which have finite 
rank at most $k$.

A quantum state is called pure if it can be represented by a single unit vector.
In quantum information theory and geometric quantum mechanics it is common to make use of the fact that mixed states can be considered as reduced pure states \cite{Nielsen_etal2010,Bengtsson_etal2008}. 
In this paper, we will use the fact that if $\KK$ is a $k$-dimensional Hilbert space and $\S(\KK,\HH)$ is the unit sphere in the space $\L(\KK,\HH)$ of linear operators from $\KK$ to $\HH$, equipped with the Hilbert-Schmidt inner product, then 
$\S(\KK,\HH)\longrightarrow\D_k(\HH),\psi\mapsto\psi\psi^\dagger$ is a surjective map.

A density operator whose evolution is governed by a von Neumann equation will remain in a single orbit for the left conjugation action of the 
unitary group of $\HH$ on $\D(\HH)$. The orbits are in bijective correspondence with the possible spectra of density operators on 
$\HH$.
By the spectrum of a density operator with $k$-dimensional support we mean the pair of sequences 
\begin{equation}\label{spectrum}
\sigma=(p_1, p_2, \dots, p_l; m_1, m_2 , \dots, m_l),
\end{equation}
where the $p_j$ are the density operator's different \emph{positive} eigenvalues, listed in descending order, and the $m_j$ are the eigenvalues' multiplicities. Throughout the 
rest of this paper we fix such a spectrum $\sigma$ and write $\D(\sigma)$ for the corresponding $\U(\HH)$-orbit in $\D_k(\HH)$. Furthermore, we 
fix an orthonormal \emph{computational basis} $\{\ket{1},\ket{2},\dots,\ket{k}\}$ in $\KK$ and we define a spectral characteristic Hermitian operator $P(\sigma)$ on $\KK$ by 
\begin{equation}\label{P}
P(\sigma)=\sum_{j=1}^{l}p_j\Pi_j,\qquad \Pi_j=\sum_{i=m_1+\dots +m_{j-1}+1}^{m_1+\dots +m_j}\ketbra{i}{i}.
\end{equation}

\subsection{Bundles of purifications over orbits of isospectral density operators}
The real part and the imaginary part of the Hilbert-Schmidt inner product define a Riemannian metric and a symplectic form on $\L(\KK,\HH)$:
\begin{equation}\label{Gustaf Otto}
G(X,Y)=\frac{1}{2}\Tr(X^\dagger Y + Y^\dagger X),\qquad \Omega(X,Y)=\frac{1}{2i}\Tr(X^\dagger Y-Y^\dagger X).
\end{equation}
Moreover, the unitary groups $\U(\HH)$ and $\U(\KK)$ act on $\L(\KK,\HH)$ from the left and from the right, respectively, by isometric and symplectic transformations:
\begin{equation}
L_U(\psi)=U\psi,\qquad R_V(\psi)= \psi V.
\end{equation}
We write $\u(\HH)$ and $\u(\KK)$ for the Lie algebras of $\U(\HH)$ and $\U(\KK)$, and $X_\xi$ and $\hat{\eta}$ for the fundamental vector fields corresponding to $\xi$ in $\u(\HH)$
and $\eta$ in $\u(\KK)$:
\begin{equation}
X_\xi(\psi)=\dd{t}\Big[L_{\exp(t\xi)}(\psi)\Big]_{t=0}=\xi\psi,\qquad \hat{\eta}(\psi)=\dd{t}\Big[R_{\exp(t\eta)}(\psi)\Big]_{t=0}=\psi\eta.
\end{equation}
Every functional on $\u(\KK)$ has the form $\mu_{\Lambda}(\xi)=i\hbar\Tr(\Lambda\xi)$ for some Hermitian operator $\Lambda$ on $\KK$.
Let $\u(\KK)^*$ be the space of all functionals on $\u(\KK)$ and define $J:\L(\KK,\HH) \to \u(\KK)^*$ by $J(\psi)=\mu_{\psi^\dagger\psi}$.
The following theorem is proved in \cite{GUR}.
\begin{theorem}
$J$ is a coadjoint-equivariant momentum map for the Hamiltonian $\U(\KK)$-action on $\L(\KK,\HH)$, and $\mu_{P(\sigma)}$ is a regular value of $J$ whose isotropy group acts freely and properly on $J^{-1}(\mu_{P(\sigma)})$.
\end{theorem}
\noindent Now let $S(\sigma)$ be the set of $\psi$ in $\L(\KK,\HH)$ satisfying $\psi^\dagger\psi=P(\sigma)$, and define 
\begin{equation}
\pi:\S(\sigma)\to\D(\sigma),\quad\pi(\psi)=\psi\psi^\dagger.
\end{equation}
The map $\pi$ is a principal fiber bundle with gauge group
$\U(\sigma)$ consisting of all unitary operators on $\KK$ which commute with $P(\sigma)$. In fact, $S(\sigma)=J^{-1}(\mu_{P(\sigma)})$, $\U(\sigma)$ is the isotropy group of $\mu_{P(\sigma)}$, and $\pi$ is canonically isomorphic to the reduced space submersion $J^{-1}(\mu_{P(\sigma)})\to J^{-1}(\mu_{P(\sigma)})/\U(\sigma)$, see \cite{GUR}. The action of $\U(\sigma)$ on $\S(\sigma)$ is induced by the 
right action of $\U(\KK)$ on $\L(\KK,\HH)$. We write $\u(\sigma)$ for the Lie algebra of $\U(\sigma)$. This algebra consists of all anti-Hermitian operators on $\KK$ which commute with $P(\sigma)$.
It follows from the Marsden-Weinstein-Meyer symplectic reduction theorem \cite{Marsden_etal1974,Meyer_1973} that $\D(\sigma)$ admits a symplectic form which is pulled back to $\Omega|_{\S(\sigma)}$ by  $\pi$. We will not need the full strength of this fact, only that $S(\sigma)$ is preserved by the left action by $\U(\HH)$. This, in turn, implies that solutions to Schr\"{o}dinger equations which extend from elements in $\S(\sigma)$ remain in $\S(\sigma)$.

\subsection{Riemannian structure and the mechanical connection}
The metric $G$ restricts to a gauge-invariant metric on $\S(\sigma)$. We define the vertical and horizontal bundles over $\S(\sigma)$ to be the 
subbundles $\V\S(\sigma)=\Ker d\pi$ and $\H\S(\sigma)=\V\S(\sigma)^\bot$
of the tangent bundle $\T\S(\sigma)$, see Figure \ref{bundle}. Here $d\pi $ is the differential of $\pi$ and $^\bot$ denotes the orthogonal complement with respect to 
$G$. Vectors in $\V\S(\sigma)$ and $\H\S(\sigma)$
are called vertical and horizontal, respectively. We equip $\D(\sigma)$ with the unique metric $g$ which makes $\pi$ a Riemannian submersion. Thus, $g$ is such that the restriction of $d\pi$ to every fiber of $\H\S(\sigma)$ is an isometry.
\begin{figure}[htbp]
\centering
\includegraphics[width=0.7\textwidth,height=0.3\textheight]{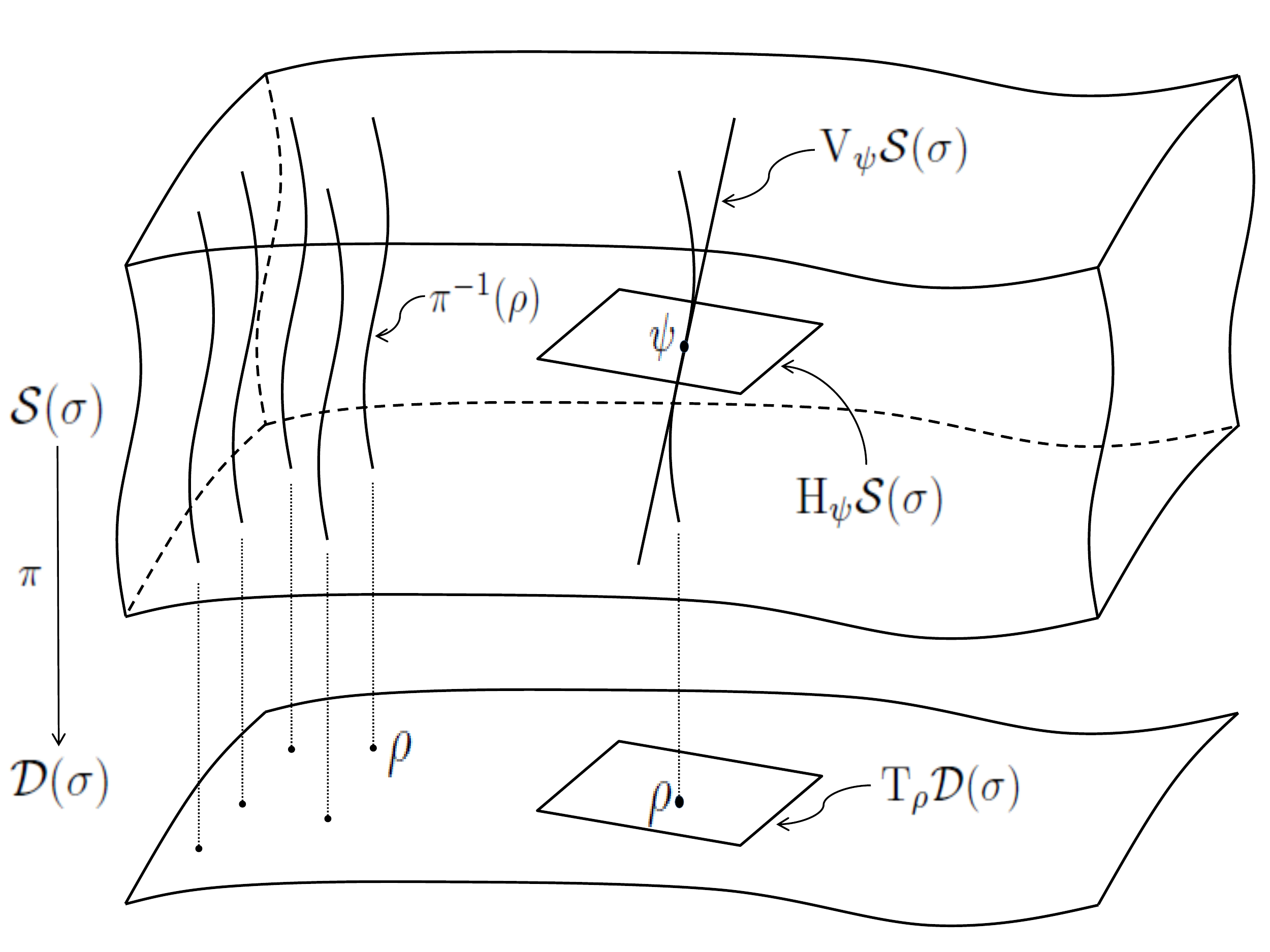}
\caption{Illustration of the bundle $\pi$ and the decomposition of each tangent space of $\S(\sigma)$ into a vertical and a horizontal subspace.}
\label{bundle}
\end{figure}

\begin{example}
If $\sigma=(1;1)$, the operators in $\D(\sigma)$ represent pure states. In fact, $\D(\sigma)$ is the projective space over $\HH$, $\S(\sigma)$ is the unit sphere in $\HH$, $\pi$ is the Hopf bundle, and $g$ is the Fubini-Study metric \cite{Anandan_etal1990,Kobayashi_etal1996I,Kobayashi_etal1996II}.
\end{example}

\begin{remark}
For a general spectrum, $\S(\sigma)$ is diffeomorphic to the Stiefel manifold of $k$-frames in $\HH$, see \cite{Kobayashi_etal1996I,Kobayashi_etal1996II}. However, $g$ is different from the Riemannian metric induced by the standard bi-invariant metric on $\U(\HH)$.
\end{remark}

The fundamental vector fields of the gauge group action on $\S(\sigma)$ yield canonical isomorphisms between $\u(\sigma)$ and the fibers in 
$\V\S(\sigma)$. Furthermore, $\H\S(\sigma)$ is the kernel bundle of the gauge invariant mechanical connection $\A_{\psi}=\mathbb{I}_{\psi}^{-1}\mathbb{J}_{\psi}$, where 
\begin{align}
\mathbb{I}&:\S(\sigma)\times\u(\sigma)\to \u(\sigma)^*, \hspace{10pt} \mathbb{I}_{\psi}\xi(\eta)=G(\hat\xi(\psi),\hat{\eta}(\psi)),\\
\mathbb{J}&:\T\S(\sigma)\to \u(\sigma)^*, \hspace{33pt}\mathbb{J}_{\psi}(X)(\xi)=G(X,\hat\xi(\psi)),
\end{align}
are the locked inertia tensor and metric momentum map, respectively.
The inertia tensor is of constant bi-invariant type since $\mathbb{I}_{\psi}$ is an adjoint-invariant form on $\u(\sigma)$ which is independent 
of $\psi$. Thus all $\mathbb{I}_{\psi}$:s define the same metric on $\u(\sigma)$, namely 
\begin{equation}\label{beta}
\xi\cdot\eta=-\frac{1}{2}\Tr\Big(\big(\xi \eta+\eta \xi\big)P(\sigma)\Big).
\end{equation}
This metric can be used to derive the following explicit formula for the mechanical connection:
\begin{equation}\label{connection}
\A_\psi(X)=\sum_{j=1}^l\Pi_j\psi^\dagger X\Pi_jP(\sigma)^{-1}.
\end{equation}
For details consult \cite{GUR}. The next proposition will be important in Section \ref{optimal hamiltonians}.

\begin{proposition}\label{geodesics conserved moment}
Geodesics in $\S(\sigma)$ have conserved metric momenta. Therefore, a geodesic in $\S(\sigma)$ which is initially horizontal remains horizontal.
\end{proposition}
\begin{proof}
Let $\psi=\psi(t)$ be a curve in $\S(\sigma)$ and $\xi$ be any element in $\u(\sigma)$. Then
\begin{equation}
\dd{t}\mathbb{J}_{\psi}(\dot{\psi})(\xi)
=\frac{1}{2}\dd{t}\Tr\Big(\big(\dot{\psi}^\dagger\psi-\psi^\dagger\dot{\psi}\big)\xi\Big)
=\frac{1}{2}\Tr\Big(\big(\nabla_t\dot{\psi}^\dagger\psi-\psi^\dagger\nabla_t\dot{\psi}\big)\xi\Big).
\end{equation}
Thus $\mathbb{J}_{\psi}(\dot{\psi})$ is constant if $\psi$ is a geodesic.
\end{proof}
\begin{example}
The proof of Proposition \ref{geodesics conserved moment} also shows that solutions in $\S(\sigma)$ to Schr\"{o}dinger equations with time-independent Hamiltonians have conserved metric momenta. 
For if $\psi$ satisfies $i\hbar\dot\psi=\obs{H}\psi$, where $\obs{H}$ is a time-independent Hamiltonian, then 
\begin{equation}
\frac{1}{2}\Tr\left(\left(\nabla_t\dot{\psi}^\dagger\psi-\psi^\dagger\nabla_t\dot{\psi}\right)\xi\right) 
=-\frac{1}{2\hbar^2}\Tr\left(\left(\psi^\dagger\obs{H}^2\psi-\psi^\dagger\obs{H}^2\psi\right)\xi\right)=0.
\end{equation}
\end{example}
\noindent In Section \ref{optimal hamiltonians} we restrict our study to the case when $\HH$ is finite dimensional and the density operators in $\D(\sigma)$ are invertible. Under these conditions, we characterize the Hamiltonians which transport elements of $\S(\sigma)$ along horizontal geodesics.

\section{An inequality of Mandelstam-Tamm type}\label{evolution time}
Anandan and Aharonov \cite{Anandan_etal1990} showed that the distance between two pure quantum states equals the length of that evolution curve connecting the two states which has the least average fluctuation in energy.
In this section we generalize Aharonov and Anandan's result to evolutions of quantum systems in mixed states. To be precise, we show that $1/\hbar$ times the path integral of the energy uncertainty of an evolving mixed state is bounded from below by the length of the curve traced out by the density operator of the state, and we show that 
every curve of isospectral density operators is generated by a Hamiltonian 
for which $1/\hbar$ times the uncertainty path integral \emph{equals} the curve's length. 
These observations give rise to an {\sc MT} evolution time estimate, which we compare with previously established {\sc MT} estimates involving the Bures length.

\subsection{Parallel and perpendicular Hamiltonians}
The average energy function $H$ of a Hamiltonian $\obs{H}$ on $\HH$ is defined by $H(\rho)=\Tr(\obs{H}\rho)$.
We write $X_H$ for the Hamiltonian vector field of $H$, $X_H(\rho)=[\rho,\obs{H}]/i\hbar$. 
This field has a distinguished gauge-invariant lift $X_{\obs{H}}$ to $\S(\sigma)$,
$X_{\obs{H}}(\psi)=\obs{H}\psi/i\hbar$. We say that $\obs{H}$ is \emph{parallel} at a density operator $\rho$ if $X_{\obs{H}}(\psi)$ is horizontal at some, hence every, $\psi$ in the fiber over $\rho$. Furthermore, we say that $\obs{H}$ parallel transports $\rho$ if the solution curve to the initial value Schr\"odinger equation
\begin{equation} \label{evolution_equation}
\dot{\psi}=X_{\obs{H}}(\psi),\qquad \psi(0)\in\pi^{-1}(\rho),
\end{equation}
is horizontal. We remind the reader that for every curve $\rho(t)$ in $\D(\sigma)$ and every initial value $\psi_0$ in the fiber over $\rho(0)$ there exists a unique horizontal curve $\psi(t)$ in $\S(\sigma)$ that extends from $\psi_0$ and is projected onto $\rho$, e.g. see \cite[p 69, Prop 3.1]{Kobayashi_etal1996I}. Furthermore, $\psi$ is the solution to \eqref{evolution_equation} for some, possibly time-dependent, Hamiltonian because $\U(\HH)$ acts transitively on $\S(\sigma)$. 

The locked inertia tensor can be used to measure deviation from parallelism:
Given a Hamiltonian $\obs{H}$ we define a $\u(\sigma)$-valued field $\xi_H$ on $\D(\sigma)$ by $\pi^*\xi_H=\A\circ X_{\obs{H}}$. 
Then $\xi_H\cdot\xi_H$ equals the square of the norm of the vertical part of $X_{\obs{H}}$. (Recall that $\cdot$ is the metric on $\u(\sigma)$ 
given by \eqref{beta}.) The field $\xi_H$ is intrinsic to the quantum system, and contains complete information 
about the expectation values of $\obs{H}$, c.f. \eqref{H} below.

The opposite of parallelism we call \emph{perpendicularity}. Thus, $\obs{H}$ is perpendicular at $\rho$ if $X_{\obs{H}}$ is vertical along the 
fiber over $\rho$, or equivalently, if $X_{\obs{H}}(\psi)=\psi\xi_H(\rho)$ for every lift $\psi$ of $\rho$. In this case $X_H(\rho)=0$. Note also that $\obs{H}$ is 
perpendicular at $\rho$ provided that $\rho$ represents a mixture of eigenstates of $\obs{H}$.

The precision to which the value of a Hamiltonian $\obs{H}$ can be known is quantified by its uncertainty function
$\Delta H(\rho)=\sqrt{\Tr(\obs{H}^2\rho)-\Tr(\obs{H}\rho)^2}$.
Let $\xi_H^\bot$ be the projection of $\xi_H$ on the 
orthogonal complement of the unit vector $-i\1$ in $\u(\sigma)$.
\begin{theorem}\label{qspeed}
The Hamiltonian vector field of $H$ satisfies 
\begin{equation}
\hbar^2 g(X_H ,X_H)=\Delta H^2 - \xi_H^\bot\cdot\xi_H^\bot.
\end{equation} 
In particular,
$\hbar^2 g(X_H(\rho),X_H(\rho))=\Delta H(\rho)^2$ if $\obs{H}$ is parallel at $\rho$.
\end{theorem}
\begin{proof}
Let $\psi$ be a purification of $\rho$. Then
\begin{align}
&\Tr(\obs{H}\rho)
=i\hbar\Tr(\A_\psi(X_{\obs{H}}(\psi))P(\sigma))
=i\hbar\Tr(\xi_H(\rho)P(\sigma))
=\hbar(-i\1)\cdot\xi_H(\rho),\label{H}\\
&\Tr(\obs{H}^2\rho)
=\hbar^2G(X_{\obs{H}}(\psi),X_{\obs{H}}(\psi))
=\hbar^2 g(X_H(\rho),X_H(\rho))+\hbar^2\xi_H(\rho)\cdot\xi_H(\rho).\label{C}
\end{align}
It follows that 
\begin{equation}\label{covariance}
\Delta H^2
=\hbar^2\left(g(X_H,X_H)+\xi_H\cdot\xi_H\right)-H^2\\
=\hbar^2\left(g(X_H,X_H)+\xi_H^\bot\cdot\xi_H^\bot\right).
\end{equation}
In particular, $\hbar^2g(X_H(\rho),X_H(\rho))=\Delta H(\rho)^2$ if $\xi_H(\rho)=0$.
\end{proof}
\begin{example}
For pure states, the vertical bundle is $1$-dimensional. Therefore $\xi_H^\bot=0$. It follows that $\hbar^2g(X_H,X_H)=\Delta H^2$. This is consistent with the observations made in \cite{Anandan_etal1990}.
\end{example}
\begin{example}
There is a canonical procedure for creating a parallel Hamiltonian from a given one:
Suppose $\rho=\rho(t)$ is a solution to the von Neumann equation with Hamiltonian $\obs{H}$.
Let $\psi=\psi(t)$ be any solution to the Schr\"{o}dinger equation on $\S(\sigma)$ with Hamiltonian $\obs{H}$, and $\obs{H}_{||}$ be any Hamiltonian on $\HH$ which is such that $\obs{H}_{||}(t)\psi(t)=\obs{H}(t)\psi(t)-i\hbar\psi(t)\xi_H(\rho(t))$. (This uniquely defines $\obs{H}_{||}(t)$ on the image of $\psi(t)$.)
Then $\obs{H}_{||}$ parallel transports $\rho(0)$ along $\rho$ with the same speed as $\obs{H}$ because
$\xi_{H_{||}}(\rho)=0$ and $[\obs{H}_{||},\rho]=[\obs{H},\rho]$. Indeed, the solution to the Schr\"odinger equation of $\obs{H}_{||}$ which extends from $\psi(0)$ is the gauge-shift of $\psi$ into a horizontal curve
\begin{equation}
\psi_{||}(t)=\psi(t)\ptexp\left(-\int_0^t\A_\psi(\dot{\psi})\dt\right),
\end{equation}
see Figure \ref{parallel}.
Here $\ptexp$ is the positive time-ordered exponential.
\end{example}

\begin{figure}[htbp]
\centering
\includegraphics[width=0.70\textwidth,height=0.3\textheight]{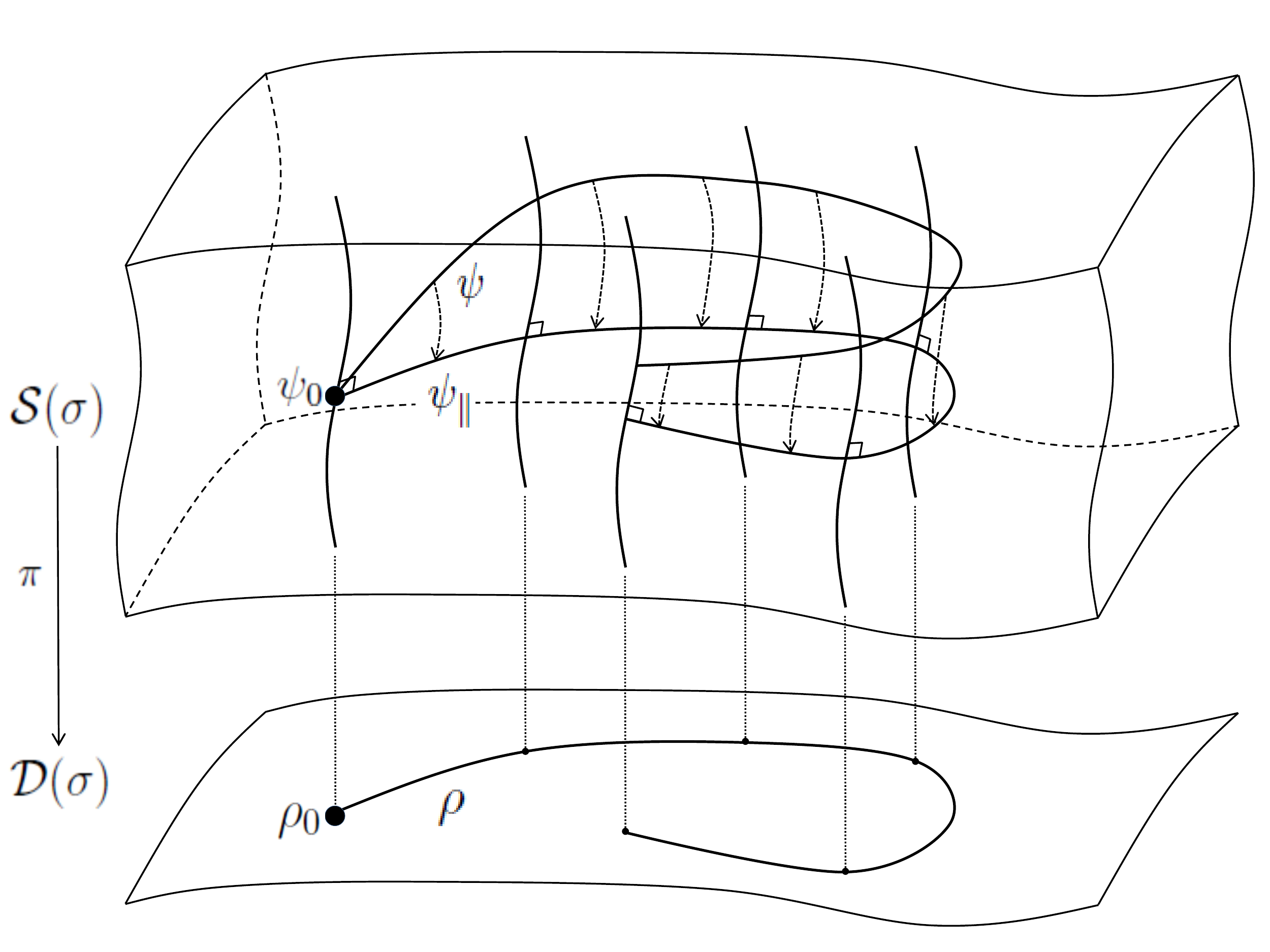}
\caption{A lift $\psi$ of an evolution curve $\rho$, and the shift of $\psi$ into a horizontal curve $\psi_{||}$.}
\label{parallel}
\end{figure}

\subsection{A Mandelstam-Tamm quantum speed limit}
The geodesic distance between two density operators with common spectrum $\sigma$ is defined as the infimum of the lengths of all curves in $\D(\sigma)$ that connect them. 
There is at least one curve whose length equals the distance, and all such curves are geodesics. Moreover, horizontal lifting of curves is length preserving because $\pi$ is a Riemannian submersion, and a curve in $\D(\sigma)$ is a geodesic if and only if one (hence all) of its horizontal lifts is a geodesic in $\S(\sigma)$, see \cite{Hermann1960}. The next theorem generalizes the main result of \cite{Anandan_etal1990}.

\begin{theorem}\label{uppskattning}
The distance between two isospectral density operators $\rho_0$ and $\rho_1$ is
\begin{equation}
\dist{\rho_0}{\rho_1}=\inf_{\obs H}\frac{1}{\hbar}\int_{0}^{\tau}\!\Delta H(\rho)\dt,\label{avstand}
\end{equation}
where the infimum is taken over all Hamiltonians $\obs H$ for which the boundary value von Neumann equation 
\begin{equation}
\dot\rho=X_H(\rho),\quad \rho(0)=\rho_0,~\rho(\tau)=\rho_1,\label{von Neumann}
\end{equation}
is solvable.
\end{theorem}
\begin{proof}
The length of a curve $\rho=\rho(t)$ in $\D(\sigma)$ is 
\begin{equation}
\length{\rho}=
\int_{0}^{\tau}\!\sqrt{g(\dot\rho,\dot\rho)}\dt.\label{length}
\end{equation}
Theorem \ref{qspeed} tells us that if $\rho$ is the integral curve of $X_H$ for \emph{some} Hamiltonian $\obs H$, the length of $\rho$ is a lower bound for the   
\emph{energy dispersion integral} along $\rho$: 
\begin{equation}
\length{\rho}\leq\frac{1}{\hbar}\int_{0}^{\tau}\!\Delta H(\rho)\dt.
\label{enekvation}
\end{equation}
There is a Hamiltonian $\obs H$ that 
generates a horizontal lift of $\rho$ because the unitary group of $\HH$ acts transitively on $\S(\sigma)$.
For such a Hamiltonian we have equality in \eqref{enekvation} by Theorem \ref{qspeed}. Moreover, if $\rho$ is a shortest geodesic, then
\begin{equation}
\dist{\rho_0}{\rho_1}=\frac{1}{\hbar}\int_{0}^{\tau}\!\Delta H(\rho)\dt.
\label{finalen}
\end{equation}
This proves \eqref{avstand}. 
\end{proof}

Aharonov and Bohm's \cite{Aharonov_etal1961} interpretation of the classic Mandelstam-Tamm time-energy uncertainty relation gives rise to a limit on the speed of dynamical evolution \cite{Bhattacharyya1983}.
For systems prepared in pure states it implies that the time it takes for a state to evolve to an orthogonal state is 
bounded from below by $\pi\hbar/2$ times the inverse of the average energy uncertainty of the system.
Uhlmann \cite{Uhlmann1992Energy} showed that the same inequality holds for mixed states when 
orthogonality is replaced by full distinguishability, which means that the fidelity of the states vanishes \cite{Englert1996,Markham_etal2008}.
To be precise, Uhlmann showed that if $\rho$ is a solution to \eqref{von Neumann},
then the Bures length between the initial and final state,
\begin{equation}
\angleB{\rho_0}{\rho_1}=\arccos\sqrt{F(\rho_0,\rho_1)},\quad F(\rho_0,\rho_1)=\left(\Tr\sqrt{\sqrt{\rho_0}\rho_1\sqrt{\rho_0}}\right)^2,
\end{equation}
is bounded from above by the energy dispersion integral:
\begin{equation}
\angleB{\rho_0}{\rho_1}\leq\frac{1}{\hbar}\int_{0}^{\tau}\Delta H(\rho)\dt.
\end{equation}
Consequently, the evolution time is bounded from below by 
$\hbar/ \Delta E$ times the Bures length:
\begin{equation}
\tau\geq\frac{\hbar}{\Delta E}\angleB{\rho_0}{\rho_1},
\qquad \Delta E=\frac{1}{\tau}\int_{0}^{\tau}\Delta H(\rho)\dt.
\end{equation}
If we combine Theorem \ref{uppskattning} with the findings of Uhlmann 
we see that 
\begin{equation}\label{estimat}
\tau\geq\frac{\hbar}{\Delta E}\dist{\rho_0}{\rho_1}\geq \frac{\hbar}{\Delta E}\angleB{\rho_0}{\rho_1}.
\end{equation}
In Section \ref{comparison} we will construct examples of density operators for which the second inequality in \eqref{estimat} is strict. Thus the quantum speed limit given by the middle term in \eqref{estimat} is sometimes greater than the quantum speed limit derived by Uhlmann. However, for fully distinguishable states they are the same: 

\begin{proposition}\label{distinkt}
If $\rho_0$ and $\rho_1$ in $\S(\sigma)$ are fully distinguishable, then 
\begin{equation}
\dist{\rho_0}{\rho_1}=\angleB{\rho_0}{\rho_1}=\pi/2.
\end{equation}
\end{proposition}

\begin{lemma}\label{orthogonal supports}
Purifications of fully distinguishable mixed states have orthogonal supports, and hence they are Hilbert-Schmidt orthogonal.
\end{lemma}
\begin{proof}
Let $\psi_0$ and $\psi_1$ be purifications of $\rho_0$ and $\rho_1$, 
and assume that $\rho_0$ and $\rho_1$ represent two fully distinguishable mixed states.
Then $\rho_0$ and $\rho_1$ have orthogonal supports, see \cite{Englert1996} and \cite[Theorem 1]{Markham_etal2008}, and 
the same is true for $\psi_0$ and $\psi_1$ because the support of $\psi_0$ equals the support of $\rho_0$, and likewise for $\psi_1$ and $\rho_1$. A compact way to express this is
$\psi_0^\dagger\psi_1=0$.
\end{proof}
\begin{proof}[Proof of Proposition \ref{distinkt}]
Let $\psi_0$ in $\pi^{-1}(\rho_0)$ and $\psi_1$ in $\pi^{-1}(\rho_1)$ be
such that the $G|_{\S(\sigma)}$-geodesic distance between them equals $\dist{\rho_0}{\rho_1}$.
If we consider 
$\psi_0$ and $\psi_1$ as elements in $\S(\KK,\HH)$, they are a distance of $\pi/2$ apart. In fact, $\psi(t)=\cos(t)\psi_0+\sin(t)\psi_1$, $0\leq t\leq \pi/2$, is a length minimizing unit speed curve from $\psi_0$ to $\psi_1$.
Consequently, $\dist{\rho_0}{\rho_1}\geq \pi/2$. However,
direct computations yield $\psi^\dagger\psi=P(\sigma)$ and $\psi^\dagger\dot\psi=0$. Thus $\psi$ is a \emph{horizontal} curve \emph{in} $\S(\sigma)$. We conclude that $\dist{\rho_0}{\rho_1}=\pi/2$.  
\end{proof}

\section{An inequality of Margolus-Levitin type}\label{margolus}
By Theorem \ref{qspeed}, the {\sc MT} quantum speed limit derived in the previous section has a geometric origin. 
This limit does not depend on the energy of the system because 
two Hamiltonians with different energies may have the same uncertainties and solution spaces for their von Neumann equations:
\begin{proposition}\label{energiinv}
Suppose $\obs{H}$ is a Hamiltonian on $\HH$. Let $E=E(t)$ be any function and define $\obs{K}=\obs{H}-E\1$. Then
 $\Delta K=\Delta H$, $X_K=X_H$, and $\xi_K^\bot=\xi_H^\bot$. But $K=H-E$. \qed
\end{proposition}
\noindent However, there is also a dynamical quantum speed limit which \emph{does} depend on the energy of the evolving system.
Margolus and Levitin \cite{Margolus_etal1998} showed that, when the evolution is governed by a time-independent Hamiltonian $\obs{H}$,
the time it takes for a system to evolve from one pure state into an orthogonal one is never less than $\pi\hbar/2(H-E_0)$, where $E_0$ is the ground state energy. The same is true for an evolution between fully distinguishable mixed states because according to Lemma \ref{orthogonal supports}, their purifications are orthogonal.
Next, we show that a similar inequality holds for a driven quantum system when $\obs{H}(s)$ and $\obs{H}(t)$ commute for $0\leq s,t\leq \tau$.
To this end, we construct a time-averaged Hamiltonian $\bar{H}$ with the same eigenspaces as $\obs{H}$ but whose eigenvalues are averages of the ground state energy shifted eigenvalues of $\obs{H}$. Thus, if
$\obs{H}(t)=\sum_nE_n(t)\ketbra{n(t)}{n(t)}$ is a continuously varying family of instantaneous spectral decompositions of $\obs{H}$,
where $\{\ket{n(t)}\}$ is an orthonormal eigenframe for $\obs{H}(t)$,
we define 
\begin{equation}
\bar{H}(t)=\sum_n\bar{E}_n(t)\ketbra{n(t)}{n(t)},\qquad \bar{E}_n(t)=\frac{1}{t}\int_0^t\left(E_n(t)-E_0(t)\right)\dt.
\end{equation}
Then we have the following generalization of the Margolus and Levitin estimate.

\begin{theorem}\label{ML}
Suppose $\rho_0$ and $\rho_1$ are the initial and final states, respectively, of an evolution $\rho$ governed by a Hamiltonian $\obs{H}$
such that $\obs{H}(s)$ and $\obs{H}(t)$ commute for $0\leq s, t\leq \tau$. 
If $\rho_0$ and $\rho_1$ are fully distinguishable, then
\begin{equation}\label{eqML}
\tau\geq \frac{\pi\hbar}{2\bar{E}},\qquad \bar{E}=\Tr(\bar{H}(\tau)\rho_0)=\Tr(\bar{H}(\tau)\rho_1).
\end{equation}
\end{theorem}
\begin{proof}
Suppose $\psi_0$ is a purification of $\rho_0$, and let $\psi=\psi(t)$ be the solution to the Schr\"{o}dinger equation with Hamiltonian $\obs{H}-E_0\1$ which extends from $\psi_0$. Then $\psi$ is a lift of $\rho$. We set $\psi_1=\psi(\tau)$.
Like Margolus and Levitin we use the inequality $\cos x\geq 2(x+\sin x)/\pi$ for $x\geq 0$ to estimate the real part of the inner product of $\psi_0$ and $\psi_1$: 
\begin{equation}\label{ett}
\begin{split}
\Re\Tr \psi_0^\dagger \psi_1
&=\sum_{n} \Re\bra{n(\tau)}\psi_1\psi_0^\dagger\ket{n(\tau)}\\
&=\sum_{n} \bra{n(\tau)}\rho_0\ket{n(\tau)}\cos\left(\frac{\tau}{\hbar}\bar{E}_n(\tau)\right)\\
&\geq \sum_{n} \bra{n(\tau)}\rho_0\ket{n(\tau)}\left(1-\frac{2}{\pi}\left(\frac{\tau}{\hbar}\bar{E}_n(\tau)+\sin \left(\frac{\tau}{\hbar}\bar{E}_n(\tau)\right)\right)\right)\\
&=1-\frac{2\tau}{\pi\hbar}\sum_{n} \bra{n(\tau)}\rho_0\ket{n(\tau)}\bar{E}_n(\tau)+\frac{2}{\pi}\Im\Tr \psi_0^\dagger \psi_1.
\end{split}
\end{equation}
Moreover,
\begin{equation}\label{tva}
\Tr(\bar{H}(\tau)\rho_0)=\sum_n \bra{n(\tau)}\rho_0\ket{n(\tau)}\bar{E}_n(\tau)
=\sum_n \bra{n(\tau)}\rho_1\ket{n(\tau)}\bar{E}_n(\tau)=\Tr(\bar{H}(\tau)\rho_1)
\end{equation}
because $ \bra{n(\tau)}\rho_0\ket{n(\tau)}=\bra{n(\tau)}\rho_1\ket{n(\tau)}$ for every $n$. Equations \eqref{ett} and \eqref{tva} yield
\begin{equation}
\Re\braket{\psi_0}{\psi_1}\geq 1-\frac{2\bar{E}}{\pi\hbar}\tau+\frac{2}{\pi}\Im\braket{\psi_0}{\psi_1},\qquad \bar{E}=\Tr(\bar{H}(\tau)\rho_0)=\Tr(\bar{H}(\tau)\rho_1).
\end{equation}
Now, by Lemma \ref{orthogonal supports}, $\psi_0$ and $\psi_1$ are orthogonal if $\rho_0$ and $\rho_1$ are fully distinguishable. Then, 
\begin{equation}
\tau\geq\frac{\pi\hbar}{2\bar{E}}.
\end{equation}
\end{proof}

\subsection{Margolus-Levitin quantum speed limit}
Margolus and Levitin's estimate has been generalized to arbitrary pairs of isospectral mixed states \cite{Giovannetti_etal2003(2)}. Recently, Deffner and Lutz \cite{Deffner_etal2013} proved 
that if $\rho_0$ and $\rho_1$ are the initial and final states of a solution curve $\rho$ to a von Neumann equation with 
Hamiltonian $\obs{H}$ such that $\obs{H}(s)$ and $\obs{H}(t)$ commute for $0\leq s,t\leq \tau$, then
\begin{equation}\label{DL-estimate}
\tau\geq \frac{4\hbar}{\pi^2\langle H-E_0\rangle}\angleB{\rho_0}{\rho_1}^2,
\qquad 
\langle H-E_0\rangle=\frac{1}{\tau}\int_0^\tau \left(\Tr(\obs{H}\rho)-E_0\right)\dt.
\end{equation}
We show how a minor modification of the proof of Theorem \ref{ML} gives rise to a similar estimate, which actually provides a greater lower bound on the evolution time of systems with time-independent Hamiltonians than \eqref{DL-estimate}.
\begin{theorem}\label{DL-estimatet}
Suppose $\rho_0$ and $\rho_1$ are the initial and final states, respectively, of an evolution governed by a Hamiltonian $\obs{H}$ such that $\obs{H}(s)$ and $\obs{H}(t)$ commute for $0\leq s,t\leq \tau$.
Let $\beta\approx 0.724$ be such that $1-\beta x$ is a tangent line to $\cos x$, see Figure \ref{subfig-1}. Then
\begin{equation}
\tau\geq \frac{4\hbar}{\beta\pi^2\bar{E}} \angleB{\rho_0}{\rho_1}^2,\qquad \bar{E}=\Tr(\bar{H}(\tau)\rho_0)=\Tr(\bar{H}(\tau)\rho_1).
\end{equation}
\end{theorem}
\begin{remark}
$\bar{E}=\langle H - E_0\rangle=H-E_0$ if $\obs{H}$ is time-independent.
\end{remark}
\begin{figure}[htbp]
\centering
    \subfloat[Graphs of $\cos x$ (solid) and $1-\beta x$ (dashed).\label{subfig-1}]{%
      \includegraphics[width=0.48\textwidth]{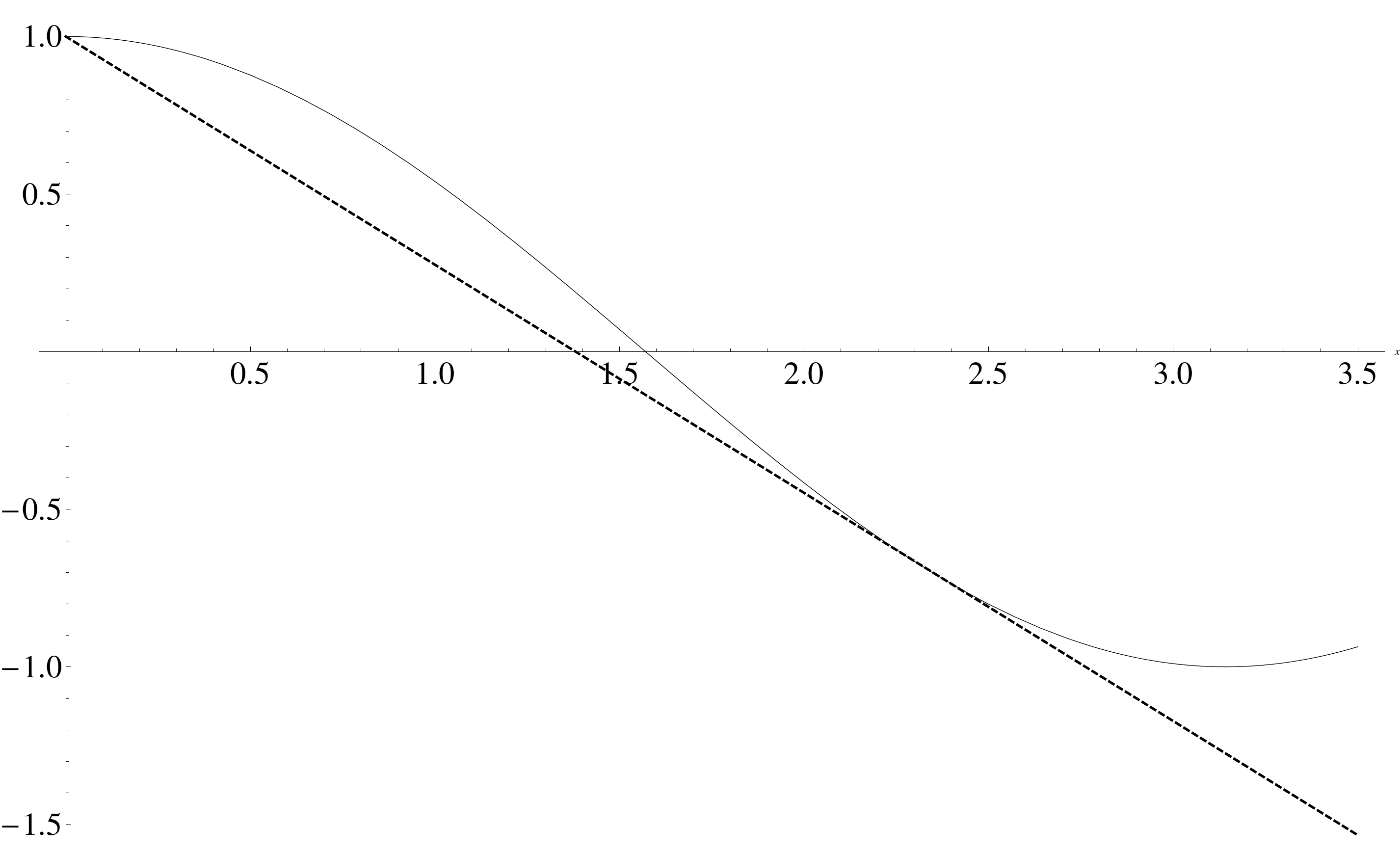}
    }
    ~
    \subfloat[Graphs of $1-x$ (solid) and $4\arccos^2x/\pi^2$ (dashed).\label{subfig-2}]{%
      \includegraphics[width=0.48\textwidth]{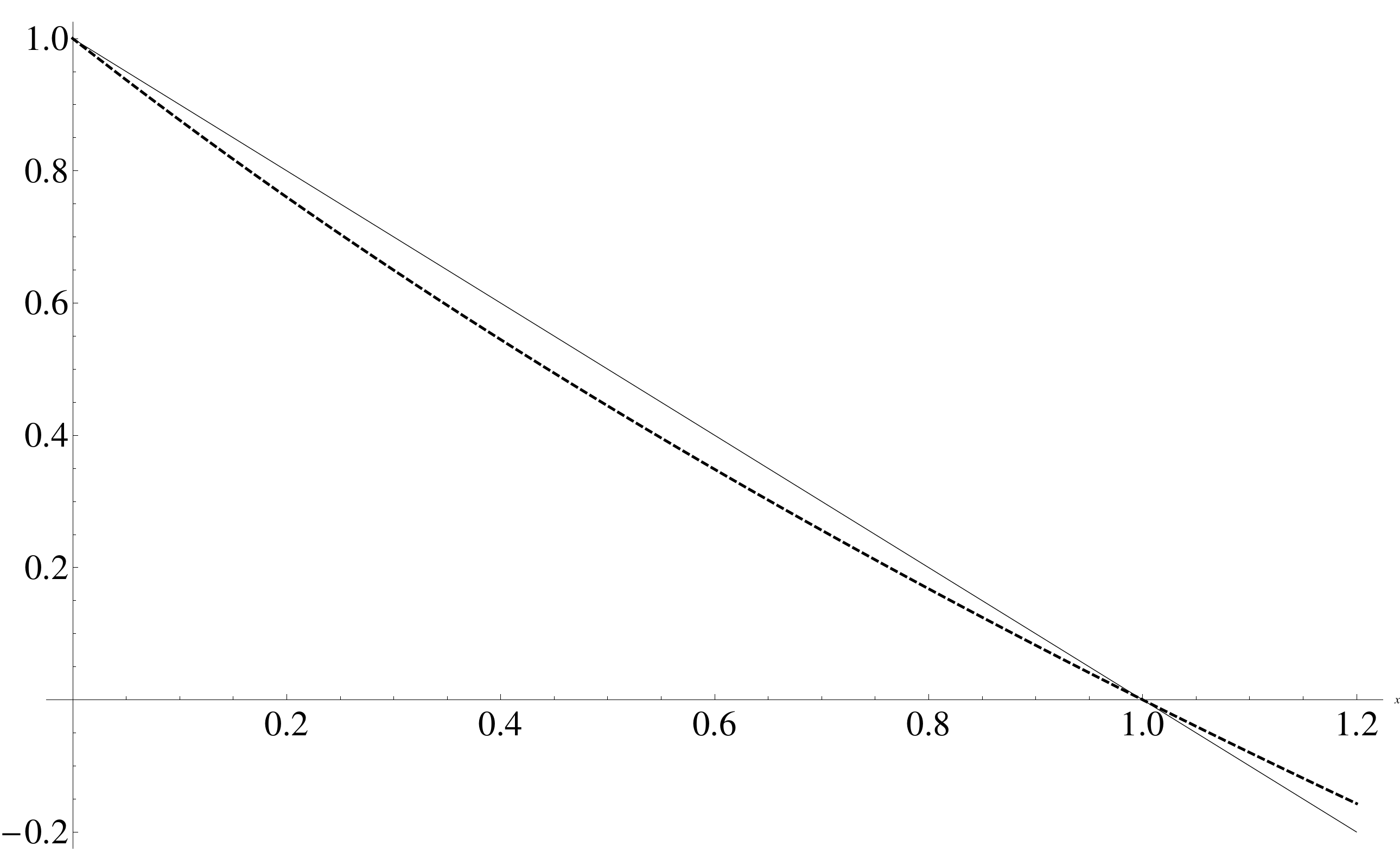}
    }
    \caption{Graphs of functions used for estimates in the proofs of Theorems \ref{ML} and \ref{DL-estimatet}.}
    \label{fig}
\end{figure}
  
\begin{proof}
Let $\psi_0$ and $\psi_1$ be as in the proof of Theorem \ref{ML}. Since $\cos x\geq 1-\beta x$ for $x\geq 0$, see Figure \ref{subfig-1}, we have that
\begin{equation}
|\Tr\psi_0^\dagger\psi_1|
\geq \Re\Tr\psi_0^\dagger\psi_1
\geq \sum_n \bra{n(\tau)}\rho_0\ket{n(\tau)}(1-\beta\tau\bar{E}_n(\tau)/\hbar)
=1-\beta\tau\bar{E}/\hbar,
\end{equation}
and since $1-x\geq 4\arccos^2 x/\pi^2$ for $0\leq x\leq 1$, see Figure \ref{subfig-2},
\begin{equation}\label{DL-improved}
\tau\geq \frac{\hbar}{\beta\bar{E}}(1-|\Tr\psi_0^\dagger\psi_1|)\geq \frac{4\hbar}{\beta\pi^2\bar{E}}\arccos^2|\Tr\psi_0^\dagger\psi_1|\geq \frac{4\hbar}{\beta\pi^2\bar{E}}\angleB{\rho_0}{\rho_1}^2.
\end{equation}
\end{proof}
\section{Optimal Hamiltonians}\label{optimal hamiltonians}
Generically, the number of independent kets in a mixed state of a finite-level quantum system equals the dimension of the Hilbert space. Such mixed states are represented by invertible density operators.
From now on we assume that $\HH$ has finite dimension $n$, and that the density operators in $\D(\sigma)$ are invertible.
Then we can put $\KK=\HH$ and express all density operators as matrices with respect to the computational basis.

\subsection{Optimal Hamiltonians for  states represented by invertible density operators}
Inspired by Theorem \ref{qspeed}, we call a Hamiltonian $\obs{H}$ \emph{optimal} for a density operator $\rho_0$ if
the solution $\rho$ to the von Neumann equation of $\obs{H}$ with initial state $\rho_0$ is a geodesic and $\xi_H^\bot$ vanishes along $\rho$.
For optimal Hamiltonians we have that $\tau=\hbar\dist{\rho_0}{\rho_1}/\Delta E$ -- at least
if the distance between $\rho_0$ and $\rho_1=\rho(\tau)$ is smaller than the injectivity radius of $\D(\sigma)$.
This follows directly from Theorem \ref{qspeed}.
In this section we characterize the optimal Hamiltonians.
Recall that a curve in $\D(\sigma)$ is a geodesic if and only if its horizontal lifts are geodesics.
Recall also that all the horizontal lifts of a given curve in $\D(\sigma)$ satisfy the same Schr\"odinger equation.
According to the next proposition (and Proposition  \ref{energiinv}), we need only characterize those optimal Hamiltonians that are parallel.
\begin{proposition}
A Hamiltonian $\obs{H}$ is optimal if and only if $\obs{H}-H\1$ is parallel and optimal.
\end{proposition}
\begin{proof}
The solution spaces of the von Neumann equations of $\obs{H}$ and $\obs{H}_{||}=\obs{H}-H\1$ are identical, and
$\xi_{H_{||}}=\xi_H^\bot$ by \eqref{H}.
\end{proof}

To each curve $\xi=\xi(t)$ in  $\u(\HH)$ we associate the Hamiltonian 
\begin{equation}
\obs{H}_\xi(t)=i\hbar\ntexp\left(\int_0^t\!\xi(t)\dt\right)\xi(t) \ptexp\left(-\int_0^t\!\xi(t)\dt\right),
\end{equation}
where $\ntexp$ and $\ptexp$ denote the negative and positive time-ordered exponentials.
The Schr\"{o}dinger equation with Hamiltonian $\obs{H}_\xi$ and initial value $\psi_0$ is
\begin{equation}\label{schrodinger}
i\hbar\dot{\psi}=\obs{H}_\xi\psi,\qquad\psi(0)=\psi_0.
\end{equation}
We describe conditions for $\xi$ that, when satisfied, makes $\obs{H}_\xi$ transport $\psi_0$, and hence every purification of $\rho_0=\psi_0\psi_0^\dagger$, along a horizontal geodesic. Recall that Proposition \ref{geodesics conserved moment} guarantees that a geodesic in $\S(\sigma)$ remains horizontal if it is initially horizontal. 

The left action by $\U(\HH)$ on $\S(\sigma)$ is free and transitive.
Therefore, the fundamental vector fields of this action define an isomorphism $\xi\mapsto X_{\xi}(\psi_0)$ from $\u(\HH)$ to the tangent space of $\S(\sigma)$ at $\psi_0$.  
We equip $\u(\HH)$ with the metric $\xi\ast\eta$ that makes this isomorphism an isometry:
\begin{equation}
\xi\ast\eta=-\frac{1}{2}\Tr((\xi\eta+\eta\xi)\rho_0).
\end{equation}
Furthermore, we write $\Lambda_\xi$ for the left invariant vector field on $\S(\sigma)$ which coincides with $X_\xi$ at $\psi_0$. Thus if $\psi=U\psi_0$, then
\begin{equation}\label{vanster}
\Lambda_\xi(\psi)=\Lambda_\xi(L_U(\psi_0))=dL_U(X_\xi(\psi_0))=U\xi\psi_0=U\xi U^\dagger\psi.
\end{equation}
To each curve $\psi=\psi(t)$ in $\S(\sigma)$, we can associate a curve $\xi=\xi(t)$ in $\u(\HH)$ by declaring $\dot{\psi}=\Lambda_\xi(\psi)$, and $\psi$  solves \eqref{schrodinger} if it extends from $\psi_0$:
\begin{proposition}\label{integral}
A curve extending from $\psi_0$ is an integral curve of $\Lambda_\xi$ if and only if it satisfies the Schr\"{o}dinger equation \eqref{schrodinger}.
\end{proposition}
\begin{proof}
Suppose $\psi=\psi(t)$ is the integral curve of $\Lambda_\xi$ that extends from $\psi_0$.
There is a unique curve of unitaries $U=U(t)$ such that $\psi=U\psi_0$. By \eqref{vanster} and the fact that $\psi_0$ is invertible, this curve satisfies $\dot U=U\xi$ and $U(0)=\1$. Thus $U(t)=\ntexp\big(\int_0^t\xi\dt\big)$. Now,
\begin{equation}
i\hbar\dot{\psi}=i\hbar\dot U\psi_0=i\hbar U\xi\psi_0=i\hbar U\xi U^\dagger U\psi_0=i\hbar U\xi U^\dagger\psi= \obs{H}_\xi\psi.
\end{equation} 
The opposite implication follows from the uniqueness of solutions to \eqref{schrodinger}.
\end{proof}
A curve $\xi=\xi(t)$ in $\u(\HH)$ satisfies the \emph{Euler-Arnold equation} if $\dot{\xi}=\ad_\xi^*\xi$, 
where $\ad_\xi^*\xi$ is the unique element in $\u(\HH)$ such that $\ad_\xi^*\xi\ast\eta=\xi\ast[\xi,\eta]$ for every $\eta$ in $\u(\HH)$. 
According to the next proposition, integral curves of $\Lambda_\xi$ are geodesics if and only if $\xi$ satisfies the Euler-Arnold equation:
\begin{proposition}\label{euler-arnold}
Suppose $\xi=\xi(t)$ is a curve in $\u(\HH)$, and let $\psi=\psi(t)$ be an integral curve of $\Lambda_{\xi}$.
Then $\psi$ is a geodesic if and only if $\xi$ satisfies the Euler-Arnold equation.
\end{proposition}

\begin{proof}
The covariant derivative of the velocity field of $\psi$ is  
$\nabla_{t}\dot \psi=\Lambda_{\dot\xi}(\psi)+\nabla_{\Lambda_\xi}\Lambda_{\xi}(\psi)$.
By the Kozul formula \cite[p 160, Prop 2.3]{Kobayashi_etal1996I},
\begin{equation}
\begin{split}
2G(\nabla_{\Lambda_\xi}\Lambda_\xi,\Lambda_\eta)
&=\Lambda_{\xi}G(\Lambda_{\xi},\Lambda_{\eta})+\Lambda_{\xi}G(\Lambda_{\eta},\Lambda_{\xi})-\Lambda_{\eta}G(\Lambda_{\xi},\Lambda_{\xi})\\
&\qquad-G(\Lambda_\xi,[\Lambda_\xi,\Lambda_\eta])+G(\Lambda_\xi,[\Lambda_\eta,\Lambda_\xi])+G(\Lambda_\eta,[\Lambda_\xi,\Lambda_\xi])\\
&=-\xi\ast[\xi,\eta]+\xi\ast[\eta,\xi]\\
&=-2\ad_\xi^*\xi\ast\eta\\
&=-\,2G(\Lambda_{\ad_\xi^*\xi},\Lambda_\eta)
\end{split}
\end{equation}
for every $\eta$ in $\u(\HH)$. Thus $\nabla_{t}\dot{\psi}= \Lambda_{\dot\xi-\ad_\xi^*\xi}(\psi)$. 
\end{proof}
\noindent The following theorem, which follows from Propositions \ref{geodesics conserved moment}, \ref{integral}, and \ref{euler-arnold},  summarizes the conditions under which the Hamiltonian $\obs{H}_\xi$ transports $\psi_0$ along a horizontal geodesic.
\begin{theorem}\label{geodesics}
Every curve in $\S(\sigma)$ that extends from $\psi_0$ is the solution to  \eqref{schrodinger} for some curve $\xi$.
Moreover, the solution to  \eqref{schrodinger} is a horizontal geodesic if and only if $\xi$ satisfies the Euler-Arnold equation, and the fundamental vector field of $\xi(0)$ is horizontal at $\psi_0$.\qed
\end{theorem}

\subsection{Density operators with two distinct eigenvalues, and almost pure qubit systems}\label{comparison}
A geodesic orbit space is a Riemannian homogeneous space 
in which each geodesic is an orbit of a one-parameter subgroup of its isometry group.
If $\sigma$ contains precisely two different, possibly degenerate, eigenvalues, then $\D(\sigma)$ is a geodesic orbit space because 
every geodesic is then generated by a time-independent Hamiltonian. 
We verify this for the geodesics extending from a density operator which is diagonal with respect to the computational basis. The general result then follows from the fact that the conjugation action of $\U(\HH)$ on $\D(\sigma)$ is transitive and by isometries. 

Assume $\sigma=(p_1,p_2;m_1,m_2)$, where $m_1+m_2=n$.
Let $\xi$ be an anti-Hermitian operator on $\HH$ such that $X_\xi$ is horizontal along the fiber over $\rho_0=P(\sigma)$, where $P(\sigma)$ is the density operator defined in \eqref{P}. Further, let 
$\eta$ be any anti-Hermitian operator on $\HH$, and express $\xi$ and $\eta$ as matrices with respect to the computational basis:
\begin{equation}
\xi=\begin{bmatrix} 0 & \xi_{12}\\ -\xi_{12}^\dagger & 0 \end{bmatrix},
\qquad
\eta=\begin{bmatrix} \eta_{11} & \eta_{12}\\ -\eta_{12}^\dagger & \eta_{22}\end{bmatrix}.
\end{equation}
Here, $\xi_{12}$ and $\eta_{12}$ have dimensions $m_1\times m_2$, and $\eta_{11}$ and $\eta_{22}$ have dimensions $m_1\times m_1$ and $m_2\times m_2$, respectively. Now $\ad_\xi^*\xi=0$ because
\begin{equation}
\xi\ast [\xi,\eta]
=\frac{1}{2}\left(p_1\Tr[\xi_{12}\xi_{12}^\dagger,\eta_{11}] + p_2\Tr[\xi_{12}^\dagger\xi_{12},\eta_{22}]\right)=0,
\end{equation}
as commutators of matrices have vanishing trace. This in turn implies that every \emph{curve} $\xi$ which satisfies the conditions in Theorem \ref{geodesics} is stationary, and hence that $\obs{H}_\xi$ is time-independent.
Next we use this observation to produce density operators representing mixed qubit states for which the second inequality in  \eqref{estimat} is strict.

Assume $\dim\HH=2$. Two independent qubits are represented by the computational basis vectors $\ket{1}$ and $\ket{2}$.
Consider an ensemble of qubits prepared so that the proportion of qubits in state $\ket{j}$ is $p_j$, where $p_1>p_2>0$.
The initial state of the ensemble is represented
by the density operator $\rho_0=\diag(p_1,p_2)$. 
Chose $\psi_0=\diag(\sqrt{p_1},\sqrt{p_2})$ in the fiber over $\rho_0$,
and let $\xi$ be an arbitrary anti-Hermitian operator on $\HH$ such that $X_\xi(\psi_0)$ is horizontal:
\begin{equation}
\xi=\begin{bmatrix} 0 & ae^{i\theta}\\ -ae^{-i\theta} & 0\end{bmatrix},\qquad a>0.
\end{equation}
The solution to the Schr\"odinger equation of $\obs{H}_\xi=i\hbar\xi$ which extends from $\psi_0$ is
\begin{equation}
\psi(t)=\begin{bmatrix} \sqrt{p_1}\cos{at} & \sqrt{p_2}e^{i\theta}\sin{at}\\ -\sqrt{p_1}e^{-i\theta}\sin{at} & \sqrt{p_2}\cos{at}\end{bmatrix}.
\end{equation}
This curve is a horizontal geodesic, and its projection is a geodesic extending from $\rho_0$: 
\begin{equation}
\rho(t)=\begin{bmatrix} p_1\cos^2{at}+p_2\sin^2{at} & e^{i\theta}(p_2-p_1)\cos {at}\sin {at}\\ e^{-i\theta}(p_2-p_1)\cos{at}\sin{at} & p_1\sin^2{at}+p_2\cos^2{at}\end{bmatrix}.
\end{equation}
Set $\rho_1=\rho(\tau)$, let $d>0$ be the (spectrum dependent) injectivity radius of $\D(\sigma)$, and assume that $0<\tau<d/a$.
Then $\rho$ is a shortest geodesic between $\rho_0$ and $\rho_1$, and
\begin{equation}
\dist{\rho_0}{\rho_1}=\length{\rho}=\frac{1}{\hbar}\int_0^\tau\Delta H_\xi(\rho)\dt= a\tau.
\end{equation}
Next, we will argue that the Bures length between $\rho_0$ and $\rho_1$ is strictly less than $ a\tau$, and hence that there exists states for which the second inequality in \eqref{estimat} is strict.

By \cite[p 225, Eq (9.47)]{Bengtsson_etal2008}, the fidelity of $\rho_0$ and $\rho_1$ is 
\begin{equation}\label{trogen}
F(\rho_0,\rho_1)
=\Tr(\rho_0\rho_1)+2\sqrt{\det\rho_0\det\rho_1}
=(p_1-p_2)^2\cos^2 (a\tau)+4p_1p_2.
\end{equation}
Therefore, the Bures length between $\rho_0$ and $\rho_1$ is 
\begin{equation}
\angleB{\rho_0}{\rho_1}
=\arccos\sqrt{(p_1-p_2)^2\cos^2(a\tau)+4p_1p_2}.
\end{equation}
The difference $a\tau-\arccos\sqrt{(p_1-p_2)^2\cos^2 (a\tau)+4p_1p_2}$ is a positive function of $\tau>0$.
(In Figure \ref{plot} we have plotted the difference for three different spectra, when $0<a\tau<\pi$. Note, however, that $d$ might be smaller than $\pi$.)
Consequently, $\dist{\rho_0}{\rho_1}>\angleB{\rho_0}{\rho_1}$. 
\begin{figure}[htbp]
\centering
\includegraphics[width=0.70\textwidth]{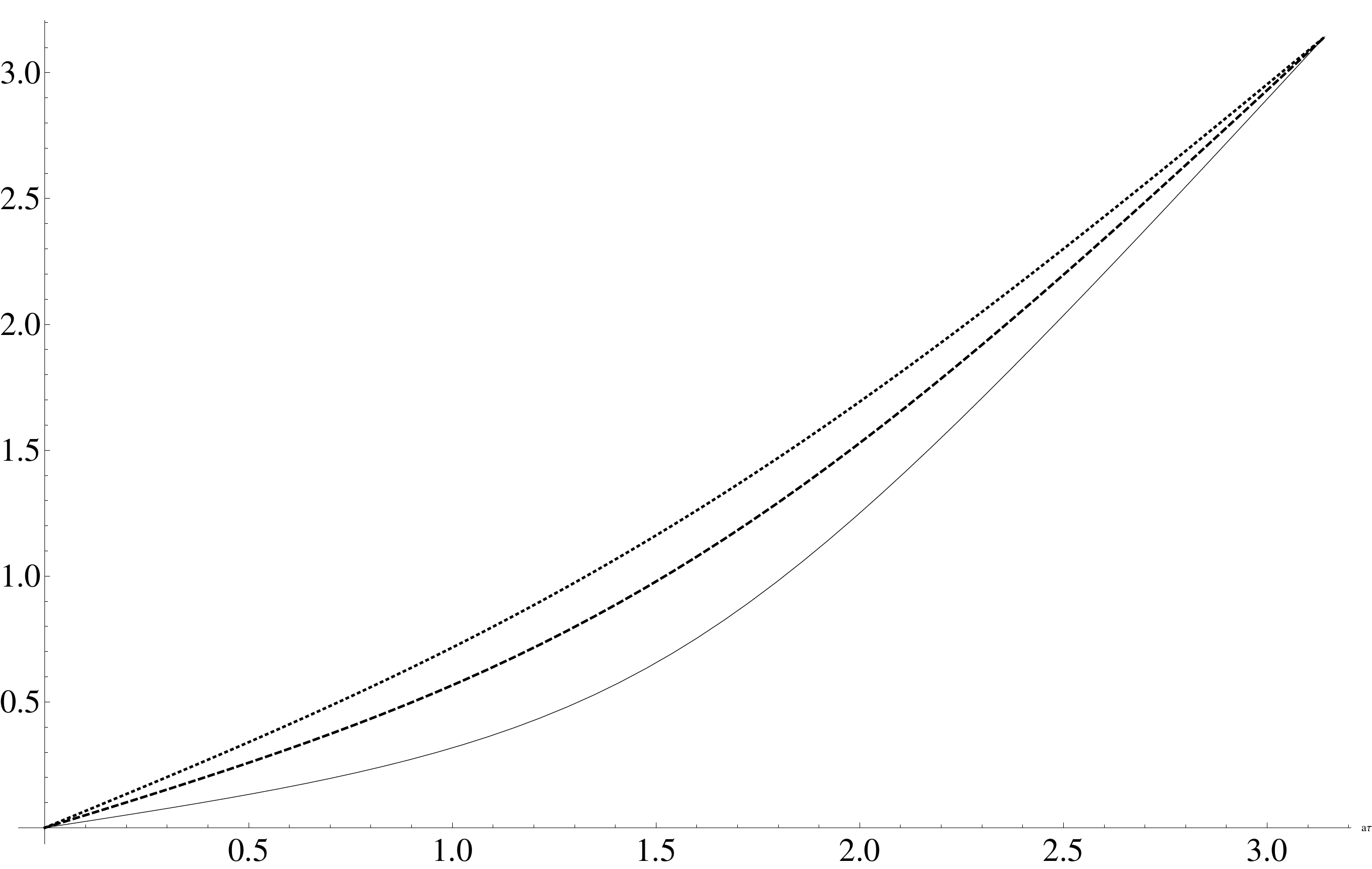}
\caption{Graph of $a\tau-\arccos\sqrt{(p_1-p_2)^2\cos^2(a\tau)+4p_1p_2}$, as a function of $a\tau$, for
 $(p_1,p_2)=(2/3,1/3)$ (dotted), $(p_1,p_2)=(3/4,1/4)$ (dashed), and $(p_1,p_2)=(7/8,1/8)$ (solid).}
\label{plot}
\end{figure}

\section{Conclusion}
Quantum speed limits are fundamental lower bounds on the time required for a quantum systems to evolve from one state into another.
In this paper, we have derived a sharp Mandelstam-Tamm quantum speed limit, by differential geometric methods, and we have characterized the Hamiltonians that optimize evolution time for finite-level quantum systems in generic mixed states.
The paper also contains a quantum speed limit of Margolus-Levitin type, which, under certain circumstances, such as that the Hamiltonian is time-independent, is sharper than those known previously. 

Quantum speed limits for open quantum systems are also available \cite{Taddei_etal2013,delCampo_etal2013,Deffner_etal2013open}. It is the intention of the authors to develop differential geometric methods by which one can derive quantum speed limits also for open systems.

\end{document}